\documentclass[12pt]{article}
\usepackage{amssymb, amsmath,cite}    
\usepackage{amsthm,bm}
\usepackage{xcolor}
\usepackage{enumerate}
\usepackage[labelfont={color=black,footnotesize},textfont={color=black,footnotesize,it,up}]{caption} 
\usepackage{url}
\usepackage[titletoc,toc,title]{appendix}
\theoremstyle{definition}
\newtheorem{lemma}{Lemma}[section]

\newtheorem{thm}{Theorem}[section]
\newtheorem{remark}{Remark}[section]
\newtheorem{Def}{Definition}

\newcommand{\norm}[1]{\left\lVert#1\right\rVert}
\newcommand{\abs}[1]{\left|#1\right|}

\usepackage{subfig,tikz}
\usetikzlibrary{decorations.pathmorphing}
\usepackage{graphics}

\newcommand{\Tr}{\,{\rm Tr}}

\newcommand{\td}{\text{d}}
\def\be{\begin{equation}}
\def\ee{\end{equation}}
\def\bea{\begin{eqnarray}}
\def\eea{\end{eqnarray}}

\title{\bf{Proof of the local mass-angular momenta inequality for  $U(1)^2$ invariant black holes }}

\author{Aghil Alaee\footnote{aak818@mun.ca } \,\,and  Hari K. Kunduri\footnote{hkkunduri@mun.ca } \\ \\
\small \sl $^a$ Department of Mathematics and Statistics, \\  \small \sl Memorial University of Newfoundland \\ \small \sl St John's NL A1C 4P5, Canada}

\date{}

\begin{document}

\maketitle






\begin{abstract}
We consider initial data for extreme vacuum asymptotically flat black holes with $\mathbb{R} \times U(1)^2$ symmetry. Such geometries are critical points of a mass functional defined for a wide class of asymptotically flat, `$(t-\phi^i)$' symmetric maximal initial data for the vacuum Einstein equations. We prove that the above extreme geometries are local minima of mass amongst nearby initial data (with the same interval structure) with fixed angular momenta.  Thus the ADM mass of nearby data $m\geq f(J_1,J_2)$ for some function $f$ depending on the interval structure.  The proof requires that the initial data of the critical points satisfy certain conditions that are satisfied by the extreme Myers-Perry and extreme black ring data. 
\end{abstract}

\section{Statement of the main result} 
Dain has proven the inequality $m \geq |J|$ for complete, maximal, asymptotically flat axisymmetric vacuum initial data to the 3+1 dimensional Einstein equation. Here $m$ is the ADM mass associated with the data and $J$ is the conserved angular momenta associated with the $U(1)$ isometry \cite{dain2006proof,Dain:2005vt ,dain2008proof}.  A thorough account of this program with references to further generalizations can be found in the review \cite{dain2012geometric}.  A natural problem is to investigate whether these results can be generalized to higher dimensions. The area-angular momenta inequalities (see \cite{dain2012geometric} for a survey) have been shown to admit such a generalization in all dimensions $D$ for black holes with $U(1)^{D-3}$ rotational isometries \cite{hollands2012horizon}.  Here we will focus on extending mass-angular momenta inequalities in $D=5$, as this is the only other possibility that admits asymptotically flat spacetimes with these isometries. 

In previous work \cite{alaee2014mass} we have constructed a mass functional $\mathcal{M}$ valid for a broad class of maximal, asymptotically flat, $U(1)^2$-invariant, $(t-\phi^i)$-symmetric, vacuum initial data. The mass functional evaluates to the ADM mass for this class and is a lower bound for the mass of general biaxisymmetric data.  We also showed that the critical points of this mass functional amongst this class of data are precisely the $\mathbb{R} \times U(1)^2$-invariant vacuum solutions of the five-dimensional Einstein equation.

Our result concerns the subset of stationary, biaxisymmetric data that represent maximal slices of extreme black holes. The uniqueness results of Figueras and Lucietti \cite{figueras2010uniqueness} imply that, for fixed angular momenta $J_1,J_2$ and interval structure, there is \emph{at most} one asymptotically flat extreme black hole.  We will consider the case where an extreme solution exists.  Then for a fixed structure we can write the mass of the extreme black hole as $m_{ext}= f(J_1,J_2)$ for some function $f$ which depends on the interval structure.  We have shown (under suitable conditions) that for small variations with fixed angular momenta about the extreme black hole initial data, the mass $m_{ext}$  is a minimum; that is 
\begin{equation}
m \geq f(J_1,J_2)
\end{equation} 
Note that $m$ could be the mass of a dynamic black hole.  This is shown by demonstrating that the extreme black holes are local minima of the mass functional. Of course, \emph{within} the two explicitly known families of stationary black holes, the extreme Myers-Perry \cite{Myers1986} and extreme doubly-spinning black ring \cite{pomeransky2006black} for fixed angular momenta, the extreme member of the family has the minimum mass, as is the case for Kerr. However, for more general interval structure,  there is no reason to expect this to occur, or indeed that a non-extreme family of solutions with a given interval structure contains an extreme limit. 

We will consider maximal initial data sets for the Einstein vacuum equations that consist of a triple $(\Sigma,h_{ab},K_{ab})$  where $\Sigma$ is complete, simply connected  Riemannian manifold with two asymptotic ends, $h_{ab}$ is a Riemannian metric , and $K_{ab}$ is a trace-free symmetric tensor field  which satisfies the vacuum  constraints
\begin{eqnarray}\label{eq:constraints}
R_h= K^{ab}K_{ab}\qquad
\bm{\nabla}^bK_{ab}= 0
\end{eqnarray}
where $R_h$ and $\bm{\nabla}$ are the scalar curvature and Levi-Civita connection with respect to $h_{ab}$. Let $m_i$ be Killing vectors generating the $U(1)^2$ symmetry of the data.  We have $\mathcal{L}_{m_i} h_{ab} = \mathcal{L}_{m_i} K_{ab} = 0$. We consider the class of metrics of the form
\begin{equation}\label{htphi}
h_{ab}=e^{2v}\tilde{h}_{ab}\qquad \tilde{h}_{ab}=e^{2U}\left(\td\rho^2+\td z^2\right)+\lambda'_{ij}\td\phi^i\td\phi^j
\end{equation} 
where $U = U(\rho,z)$ is a smooth function, $\lambda'=[\lambda'_{ij}]$ is a positive definite $2\times 2$ symmetric metric with $\det\lambda'=\rho^2$ and $\phi^i$ are coordinates with periodicity $2\pi$ adapted to the Killing vectors $m_i$.  Note that we assume that the action of the  $U(1)^2$ isometry is orthogonally transitive. We expect that this assumption can be removed~\cite{dain2008proof}.  (Of course, if the data arises from a \emph{stationary} spacetime, this assumption can be removed). 

In the following we will \emph{not} assume the data is $t-\phi^i$ symmetric. Rather, we restrict attention to metrics of the form \eqref{htphi} but we allow for general axisymmetric extrinsic curvature. As has been proved in \cite{alaee2014mass}, one can always decompose $K_{ab}$ as
\begin{equation} \label{decomposition}
K_{ab} = \mathcal{K}_{ab} + H_{ab}
\end{equation} where $\mathcal{K}_{ab}$ is the $t-\phi^i$-symmetric part of the extrinsic curvature. Recall that $(t-\phi^i)$-symmetry implies that under the diffeomorphism $\phi^i \to -\phi^i$, we have $h_{ab} \to h_{ab}, \mathcal{K}_{ab} \to -\mathcal{K}_{ab}$ \cite{alaee2014mass}.

We now briefly review the construction of the mass functional which is defined for $t-\phi^i$ symmetric data $(\Sigma,h,\mathcal{K})$. Since $\mathcal{K}_{ab}$ is automatically traceless, using the divergence-less condition and the property $\Sigma$ is simply connected \cite{alaee2014small}, we can express $\tilde{K}_{ab}=e^{2v}\mathcal{K}_{ab}$ in a compact form.  Define two scalar potentials $Y^i$ and one-forms
\begin{equation}
S^i = \frac{1}{2 \det \lambda'} i_{m_1} i_{m_2} \star \td Y^i
\end{equation}  Note $\td \star S^i = 0$. Then an \emph{arbitrary} divergenceless $t-\phi^i$-symmetric extrinsic curvature can be expressed as \cite{alaee2014small}
\begin{equation}\label{exttphi}
\tilde{K}_{ab}=\frac{2}{\det\lambda'}\Bigg[\left(\lambda'_{22} S^1{}_{(a}m_1{}_{b)}-\lambda'_{12}S^2{}_{(a}m_1{}_{b)}\right)+\left(\lambda'_{11} S^2{}_{(a}m_2{}_{b)}-\lambda'_{12}S^1{}_{(a}m_2{}_{b)}\right)\Bigg].
\end{equation}  Hence for $(t-\phi^i)$ symmetric initial data, the extrinsic curvature is completely characterized by the scalar potentials $Y^i$ as well as the metric functions $\lambda'_{ij}$. One can show \cite{alaee2014small} that these potentials are simply the pull-backs of the spacetime twist potentials defined in the usual way, i.e. $\td Y^i = \star_5(m_1 \wedge m_2 \wedge \td m_i)$. Moreover, these potentials are related to the angular momenta of the data by
\begin{equation}
J_i=\frac{\pi}{4}\left[Y^i(\rho=0,z)-Y^i(\rho=0,-z)\right]
\end{equation}
In terms of the conformal data $(\tilde{h}_{ab},\tilde{K}_{ab},v)$ the constraint equations reduce to the Lichnerowiscz equation for $v$:
\begin{gather}
\Delta_{\tilde{h}}\Phi-\frac{1}{6} R_{\tilde h}\Phi  +\frac{1}{6}\tilde{K}_{ab}\tilde{K}^{ab}\Phi^{-5}=0.\label{eq:Lich}
\end{gather}  where $\Phi = e^{2v}$.
\begin{remark}\cite{alaee2014mass} Let $(\Sigma,h,\mathcal{K})$ be an asymptotically flat, $(t-\phi^i)$-symmetric, vacuum initial data set. Such data can be completely characterized by $\Sigma$, its $U(1)^2$ action, and a triple $u = (v,\lambda',Y)$ where $v$ is a scalar, $\lambda'$ is a positive definite symmetric matrix with determinant $\rho^2$, and $Y=(Y^1,Y^2)^t$ is a column vector (the function $U$ is found by solving a Poisson equation arising from \eqref{eq:Lich}). We will denote such data simply by $(\Sigma,u)$. \label{tphiremark}
\end{remark}
  
Let $\rho$, $z$, $\phi$ be cylindrical coordinates in Euclidean $\mathbb{R}^3$ with metric  $\delta_3=\td\rho^2+\td z^2+\rho^2\td\phi^2$.   Note all functions only depend on $\rho$ and $z$.  Then by \cite{alaee2014mass} we have the following mass functional defined for $(\Sigma,u)$
\begin{eqnarray}
\mathcal{M}(u)=\frac{1}{8}\int_{\mathbb{R}^3}\left(-\frac{\det\nabla\lambda'}{2\rho^2}+e^{-6v}\frac{\nabla Y^t\lambda'^{-1}\nabla Y}{2\rho^2}+6\left(\nabla v\right)^2\right)\, \, \td\Sigma-\frac{\pi}{4}\sum_{\text{rods}}\int_{I_i}\log V_i\,\td z
\end{eqnarray}
where $\td\Sigma=\rho\,\td\rho\td z\td\phi$ and $\nabla$ are respectively the volume element and connection with respect to $\delta_3$, and $V_i$ is defined by 
\begin{equation}
V_i(z)=\lim_{\rho\to 0}\frac{2\sqrt{\rho^2+z^2}\lambda'_{ij}w^iw^j}{\rho^2},\qquad z\in I_i=(a_i,a_{i+1}),\quad w^i\in\mathbb{Z}
\end{equation}
where $\lambda'_{ij}w^j = O(\rho^2)$ as $\rho \to 0$ with $w=w^i\frac{\partial}{\partial\phi_i}$ is the Killing vector vanishing on the rod $I_i$ Note that $\phi$ is an auxiliary coordinate with period $2\pi$ and the functional can be defined over the orbit space $\mathcal{B}\cong\Sigma/U(1)^2$ \cite{alaee2014mass}.  $\mathcal{B}$ is a two-dimensional manifold with boundary and corners \cite{hollands2008uniqueness} and the boundary and asymptotic conditions on the various functions which parametrize the data are given in Section II of \cite{alaee2014mass}.  We record them here for convenience.  

To understand the decay in the asymptotic regions, we define new coordinates
\begin{equation}
x \equiv \frac{z}{\sqrt{\rho^2 + z^2}}\; , \qquad r \equiv \left[2\sqrt{\rho^2+z^2}\right]^{1/2}
\end{equation}  where $x \in [-1,1]$ and $r \in (0,\infty)$. Observe that $\delta_3 = r^2(\td r^2 +\tfrac{r^2}{4}[ (1-x^2)^{-1} \td x^2 + (1-x^2)\td \phi^2])$. Note the boundary $\rho=0$ corresponds to $x= \pm 1$ and $r \to 0$ corresponds to an asymptotic end which can be either asymptotically flat or cylindrical whereas $r \to \infty$ corresponds to the asymptotically flat end where the ADM mass is defined.  We require:
\begin{enumerate}[(a)]
\item as $r\to\infty$  
\begin{gather} 
v=o_{1}(r^{-1}),\quad\lambda'_{ij}-\sigma_{ij}=\frac{f_{il}\sigma_{lj}}{r^2}+\sigma_{ij}o_1(r^{-2}),\label{Asympvf}\\
V=\frac{\bar{V}(x)}{r^2}+o_1(r^{-2}),\quad \int_{-1}^1\bar{V}(x)\,\td x=0.\label{AsympV}
\end{gather}
where  $\sigma_{ij}=\frac{r^2}{2}\text{diag}\left(1+x,1-x\right)$  and $f_{ij}$ is a diagonal matrix with $\text{Tr}(f_{ij})=0$. This implies that the geometry approaches the flat metric on $\mathbb{R}^4$ at large $r$. 
\item As $r\to 0$ for an asymptotically flat end we have
\begin{gather}
v=-2\log (r)+o_{1}(1),\quad\lambda'_{ij}-\sigma_{ij}=f_{il}\sigma_{lj}r^2+\sigma_{ij}o_1(r^{2}),\\
V=\bar{V}(x)r^2+o_1(r^{2}),\quad \int_{-1}^1\bar{V}(x)\,\td x=0.\label{End1}
\end{gather}
\item As $r\to 0$ for an asymptotically cylindrical end with topology $\mathbb{R}^+\times N$ where $N\cong S^3,S^1\times S^2, L(p,q)$ we have
\begin{gather}
v=-\log (r)+o_{1}(r^{1}),\quad\lambda'_{ij}-\bar{\sigma}_{ij}=o_1(r^{2}),\quad
V=O_1(1)\label{End2}
\end{gather}
where $h^c=e^{2V} \frac{ \td x^2}{4(1-x^2)}+\bar{\sigma}_{ij}\td\phi^i\td\phi^j$ is the metric on $N$.
\end{enumerate}

\begin{remark} \label{remark1} The mass functional is defined for $t-\phi^i$ symmetric data $(\Sigma,u)$ and it equals the ADM mass. The ADM mass of  \emph{general} initial data $(\Sigma,h,K)$ satisfies \cite{alaee2014mass}
\begin{equation}
m \geq \mathcal{M}(u)
\end{equation} where $u = (v,\lambda',Y)$ is constructed from the corresponding $\mathcal{K}_{ab}$ associated to $K_{ab}$ by the decomposition \eqref{decomposition}. The equality is achieved if and only if the original initial data set is $t-\phi^i$ symmetric. 
\end{remark}

From now on we restrict attention to the mass functional, as it is a lower bound for the mass of our original initial data.  We set $\varphi=(\bar{v},\bar{\lambda}', \bar{Y})$  where $\bar{\lambda'}$ is a symmetric $2 \times 2$ matrix such that $\det\bar{\lambda'}=0$.  As will be explained in following sections, $\varphi$ will represent a perturbation about some fixed initial data $u_0$ defined in Definition \ref{Def1} . This should consist of five free degrees of freedom, and the apparent restriction $\det\bar{\lambda'} =0$ is simply a gauge choice.  Let $\Omega$ be a (unbounded) domain and  we introduce the following weighted spaces of $C^1$ functions with norm
\begin{equation}
\norm{f}_{C^1_{\beta}(\Omega)}=\sup_{x\in \Omega}\{\sigma^{-\beta}\abs{f}+\sigma^{-\beta+1}\abs{\nabla f}\}
\end{equation} is finite
with $\beta<-1$ and $\sigma=\sqrt{r^2+1}$ and for a column vector and a matrix we define respectively
\begin{equation}
\abs{\bar{Y}} \equiv \left(\bar{Y}^t \lambda'^{-1}_0 \bar{Y}\right)^{1/2}\;, \quad
\abs{\bar{\lambda}'} \equiv \left(\text{Tr}\left[\bar{\lambda}'^t\bar{\lambda}'\right]\right)^{1/2}
\end{equation}
Let $\rho_0> 0$ be a constant and $K_{\rho_0}$ be the cylinder $\rho\leq \rho_0$ in $\mathbb{R}^3$. We define the domain $\Omega_{\rho_0}=\mathbb{R}^3 \backslash K_{\rho_0}$. The perturbation $\bar{Y}$ and $\bar{\lambda}$ are assumed to vanish in $K_{\rho_0}$. This is consistent with the physical requirement that the
perturbations keep fixed the angular momenta $J_i$ and fixed orbit space. The Banach space $B$ is defined by
\begin{equation}
\norm{\varphi}_{B}=\norm{\bar{v}}_{C^1_{\beta}(\mathbb{R}^3)}+\norm{\bar{\lambda}'}_{C^1_{\beta}(\Omega_{\rho_0})}+\norm{\bar{Y}}_{C^1_{\beta}(\Omega_{\rho_0})}
\end{equation}
Now we define the class of extreme data. Note that we will denote non-negative constants which depend on parameters of data such as mass and angular momenta by $C$, $C_i$, and $C'$. 
\begin{Def}\label{Def1}
The set of \emph{extreme class} $E$ is the collection of data arising from extreme, asymptotically flat, $\mathbb{R}\times U(1)^2$ invariant black holes which consist of triples $u_0 = (v_0,\lambda'_0,Y_0)$ where $v_0$ is a scalar, $\lambda'_0=[\lambda_{ij}]$ is a positive definite $2\times 2$ symmetric matrix, and $Y_0$ is a column vector with the following bounds for $\rho\leq r^2$
\begin{enumerate}
\item $\frac{\nabla Y_0^t\lambda^{-1}_0\nabla Y_0}{X_0}\leq Cr^{-4}$ and $e^{-2v_0}\frac{\nabla Y_0^t\lambda^{-1}_0\nabla Y_0}{X_0}\leq Cr^{-2}$ in $\mathbb{R}^3$
 where $\lambda_0=e^{2v_0}\lambda'_0$
 \item $C_1\rho I_{2\times 2}\leq\lambda_{0}\leq C_2 \rho I_{2\times 2}$ and $C_3\rho^{-1}I_{2\times 2}\leq\lambda^{-1}_{0}\leq C_4\rho^{-1}I_{2\times 2}$ in $\Omega_{\rho_0}$
\item $\rho^2\leq X_0$ in $\mathbb{R}^3$ where $X_0=\det\lambda_0$ and $X_0^2\leq C' \rho^4$ in $\Omega_{\rho_0}$ where $\lim_{\rho_0\to 0}C'=\infty$
\item $\abs{\nabla v_0}^2\leq C r^{-4}$, $\abs{\nabla\ln X_0}^2\leq C\rho^{-2}$ in $\mathbb{R}^3$ and $\abs{\nabla\lambda_0\lambda^{-1}_0}^2\leq C\rho^{-2}$ in $\Omega_{\rho_0}$ 
\end{enumerate} 
\end{Def} 
The choice of these bounds are consistent with the two known extreme black holes initial data, extreme Myers-Perry and extreme doubly spinning black ring.  These inequalities are difficult to prove directly because the expressions in terms of the $(\rho,z)$ coordinates are unwieldy. However, we have checked numerically that these bounds hold for a wide range of parameters for these two cases. It is possible that there exists an extreme data which has slightly different bounds (i.e. this would correspond to another extreme black hole with different orbit space).  In that case we expect the arguments used in the proof of theorem \ref{main theorem} can be extended to take into account these different estimates.

Note that by what has been proved in \cite{alaee2014mass}, $\mathcal{M}$ evaluated on the extreme class is non-negative and given by 
\begin{equation}\label{massextreme}
\mathcal{M}_{\text{cp}}=\frac{3}{8}\int_{\mathbb{R}^3}e^{-6v_0}\frac{\abs{\nabla Y_0}^2 }{2\rho^2}\, \,\td \Sigma
\end{equation}
 where $\abs{\nabla Y_0}^2=\nabla Y_0^t\lambda'^{-1}_{0}\nabla Y_0$. Now denote an extreme data of this class by $u_0 = (v_0,\lambda'_0,Y_0)\in E$. Then we have the following result

\begin{thm}\label{main theorem} $\phantom{Next line}$
\begin{enumerate}[(a)]
\item  Let  $\varphi=(\bar{v},\bar{\lambda}', \bar{Y}) \in B$ where $B$ is the Banach space defined above and $u_0=({v}_0,{\lambda}_0',{Y_0})\in E$ is extreme data with  fixed $\mathcal{B}$. Then the functional $\mathcal{M}:B\rightarrow \mathbb{R}$  has a strict local
minimum at $u_0$. That is, there exists $\epsilon>0$ such that
\begin{equation}
\mathcal{M}(u_0+\varphi)>\mathcal{M}(u_0)
\end{equation}
for all $\varphi\in B$ with $\norm{\varphi}_{B}<\epsilon$ and $\varphi\neq 0$.
\item  Let $(\Sigma,h_{ab},K_{ab})$ be an asymptotically flat, maximal, $U(1)^2$-invariant, vacuum initial data with mass $m$ and angular momenta $J_1$ and $J_2$ and fixed orbit space $\mathcal{B}$ such that the data satisfies the boundary conditions given by \eqref{Asympvf}-\eqref{End2}. Let $u = (v,\lambda',Y)$ describe the associated $t-\phi^i$ symmetric data as in Remark \ref{remark1} and write $u = u_0 + \varphi$ where $u_0$ is extreme data with the same $J_1,J_2$ and orbit space $\mathcal{B}$.  If $\varphi$ is sufficiently small (as in (a)) then 
\begin{equation}
 m\geq f(J_1,J_2) = \mathcal{M}(u_0)
 \end{equation} for some $f$ which depends on the orbit space $\mathcal{B}$. Moreover, $m=f(J_1,J_2)$ for data $(\Sigma,h,K)$ in a neighbourhood if and only if the data are  extreme data.
\end{enumerate}   
\end{thm}
\noindent For the sake of illustration we mention two special cases of the theorem.   
\begin{enumerate}
\item In dimension 5, a possible horizon topology is  $H\cong S^3$. Consider fixed angular momenta $J_1$  and $J_2$ and fixed orbit space $\tilde{\mathcal{B}}$ consisting of a finite timelike interval (the event horizon) and two semi-infinite spacelike intervals extending to asymptotic infinity (representing rotation axes). Then the orbit space of the slice will be $\mathcal{B} \cong \tilde{\mathcal{B}}\backslash \{\text{horizon interval}\}$ which corresponds to slice topology $\Sigma\cong\mathbb{R}\times S^3$ \cite{alaee2014notes,alaee2014mass}. By the uniqueness theorem \cite{figueras2010uniqueness} extreme Myers-Perry solution is the unique solution with this orbit space and fixed angular momenta. Thus there exists $f(x,y)=3\left[\frac{\pi}{32}(\abs{x}+\abs{y})^2\right]^{1/3}$ such that mass of extreme Myers-Perry is equal to $f(J_1,J_2)$. Then by theorem \ref{main theorem} mass of any asymptotically flat, maximal, biaxisymmetric data sufficiently close (in the sense made precise above) with the same interval structure and  angular momenta is greater than $f(J_1,J_2)$. 
\item Now consider the horizon topology $H\cong S^2\times S^1$. Consider fixed angular momenta $J_1$  and $J_2$ and fixed orbit space $\tilde{\mathcal{B}}$ consisting a finite timelike interval, a finite spatial interval, and two semi-infinite intervals extending to asymptotic infinity. Then the orbit space of the slice will be $\mathcal{B} \cong \tilde{\mathcal{B}}\backslash \{\text{horizon interval}\}$ which corresponds to slice topology $\Sigma\cong S^2\times B^2\#\mathbb{R}^4$  \cite{alaee2014notes,alaee2014mass}. By the uniqueness theorem \cite{figueras2010uniqueness} the extreme doubly spinning black ring is the unique solution with orbit space $\tilde{\mathcal{B}}$ and fixed angular momenta. Thus there exist $f(x,y)=3\left[\frac{\pi}{4}\abs{x}(\abs{y}-\abs{x})\right]^{1/3}$\footnote{In \cite{alaee2014mass} there is a typo in equation (2). The correct expression is $M^3=\frac{27\pi}{4}J_1(J_2-J_1)$} such that mass of extreme doubly spinning black rings is equal to $f(J_1,J_2)$. Then by theorem \ref{main theorem} the mass of any asymptotically flat, maximal, biaxisymmetric data with the same orbit structure and fixed angular momenta is greater than $f(J_1,J_2)$. 
\end{enumerate}  

Theorem \ref{main theorem} is a local inequality which should be satisfied for a wide class of (possibly dynamical) black holes with a fixed interval structure with a geometry sufficiently near an extreme black hole.  One may expect to prove a global result showing that this inequality holds all data with fixed $J_1,J_2$ and $\mathcal{B}$. Such a global inequality has been proved in the electrovacuum in 3+1 dimensions \cite{chrusciel2009mass,dain2008proof}.   A major obstacle to extending this result to the present case is showing positivity of $\mathcal{M}$ for arbitrary interval structures consistent with asymptotic flatness. However, for a class of interval structures (including Myers-Perry black hole initial data) one can show $\mathcal{M} \geq 0$ \cite{alaee2014mass}.  We are currently investigating whether a global inequality can be demonstrated in this particular setting.  In this context, it is worth noting that $\mathbb{R} \times U(1)^2$-invariant vacuum spacetimes can be cast as harmonic maps from the orbit space to $SL(3,\mathbb{R}) / SO(3)$ \cite{hollands2008uniqueness}. The target space metric is easily checked to be Einstein with negative curvature (it is not conformally flat). This can be contrasted with the four-dimensional case where the $\mathbb{R} \times U(1)$-invariant vacuum solutions are harmonic maps to $SL(2,\mathbb{R})/SO(2) \cong \mathbb{H}^2$ equipped with its standard Einstein metric.  

Another open problem is to generalize this theorem to include multiple asymptotic ends, corresponding to multiple black holes \cite{chrusciel2008mass}.  

The proof of theorem \ref{main theorem} is given in Section \ref{proof}. The rest of the paper is  organized as follows. In Section \ref{critical} we find critical points of $\mathcal{M}$ and we prove uniform continuity of a one parameter family of functionals obtained from $\mathcal{M}$ and denoted by $\mathcal{E}_{\varphi}(t)$. In Section \ref{Carter}  we will use a Carter-type identity (a linearized version of Mazur's identity) to derive an identity for five dimensional spacetimes and we use this identity to prove positivity of the second variation of $\mathcal{E}_{\varphi}(t)$ at $t=0$. Finally, we prove a coercive condition for the second variation $\mathcal{E}''_{\varphi}(0)$. This is sufficient to demonstrate that $u_0$ is a strict minimum for $\mathcal{M}$.  
\section{Critical points of the mass functional $\mathcal{M}$}\label{critical}
In this section we will study the properties of second variation of mass functional $\mathcal{M}$. Let $\varphi\in B$ and consider the real-value function
\begin{equation}
\mathcal{E}_{\varphi}(t)\equiv \mathcal{M}(u_0+t\varphi)
\end{equation}
and we assume 
\begin{equation}
(v,\lambda',Y)\equiv(v(t),\lambda'(t),Y(t))=(v_0+t\bar{v},\lambda'_0+t\bar{\lambda}',Y_0+t\bar{Y})\label{relation1}
\end{equation}
where $\det\lambda'=\rho^2$. This choice for determinant of $\lambda'$ requires that $\det\bar{\lambda}=0$. Moreover we have
\begin{equation}
\lambda\equiv \lambda (t)=e^{2v}\lambda'(t)\qquad X\equiv X(t)=e^{4v}\rho^2\label{relation2}
\end{equation}
and $X_0=X(0)$. Then the first variation is 
{\small\begin{eqnarray}
\mathcal{E}'_{\varphi}(t)&=&\frac{1}{8}\int_{\mathbb{R}^3}\Bigg[12\nabla v.\nabla\bar{v}+\frac{e^{-6v}}{2\rho^4}\Bigg[\nabla Y^t\text{adj}(\bar{\lambda'})\nabla Y+2\nabla Y^t\text{adj}(\lambda')\nabla\bar{Y}-6\bar{v}\nabla Y^t\text{adj}(\lambda')\nabla Y\Bigg]\nonumber\\
&-&\frac{1}{2\rho^2}\text{Tr}\left(\text{adj}(\nabla\bar{\lambda'})\nabla\lambda'\right)\Bigg]\,d\Sigma
\end{eqnarray}}
The critical points of this variation ($\mathcal{E}'_\phi(0) =0$)  in \cite{alaee2014mass}  are given by 
 \begin{eqnarray}\label{EL1}
 G_X&\equiv&4\Delta_3v+\frac{\nabla Y^t{\lambda}^{-1}\nabla{Y}}{X} = 0\\ \label{EL2}
  G&\equiv&\nabla\cdot\left(\frac{\nabla\lambda'}{\rho^2}\right)+\frac{e^{2v}}{X^2}\nabla Y\nabla Y^t = 0\\ \label{EL3}
 G_Y&\equiv&\nabla\cdot\left(\frac{\lambda^{-1}\nabla Y}{X}\right) = 0
 \end{eqnarray}
On the other hand, the vacuum field equations for a $\mathbb{R} \times U(1)^2$-invariant spacetime are \cite{figueras2010uniqueness}
\begin{gather}\label{fieldeqns}
\begin{aligned}
G_{\lambda}&\equiv \nabla\cdot \left(\lambda^{-1}\nabla\lambda\right)+\frac{\lambda^{-1}}{X}\nabla Y\cdot\nabla Y^t = 0\\
G_{Y}&=\nabla\cdot \left(\frac{\lambda^{-1}}{X}\nabla Y\right) = 0
\end{aligned}
\end{gather} 
where $G_X=\text{Tr}\left(G_{\lambda}\right)$. It is straightforward to show these field equations \eqref{fieldeqns} are equivalent to critical points (\ref{EL1})-(\ref{EL3}) of $\mathcal{E}_{\varphi}$. This shows the critical points of the mass functional are the same as the stationary, biaxisymmetric vacuum solutions \cite{alaee2014mass} (written in spacetime Weyl coordinates with orbit space $\tilde{\mathcal{B}}$). However, for non-extreme black holes, this chart only covers the exterior region of the black hole spacetime and the manifold has an interior boundary.  In particular in these coordinates the mass functional is singular on the inner boundary. One can always find quasi-isotropic coordinates on the initial data slice $\Sigma$  to complete the manifold and compute the mass, but then the resulting geometry is \emph{not} a critical point of $\mathcal{M}$.  But for extreme black holes, the usual spacetime Weyl coordinates and quasi-isotropic coordinates coincide, and the mass functional is well defined on these critical points. This point is discussed in more detail\footnote{We thank S Dain for clarifying this point.} in \cite{dain2006variational} and {\cite{alaee2014mass}}. 

A calculation yields the second variation 
{\small\begin{eqnarray}
\mathcal{E}''_{\varphi}(t)&=&\frac{1}{8}\int_{\mathbb{R}^3}\Bigg(12\left(\nabla\bar{v}\right)^2-\frac{\det\nabla\bar{\lambda'}}{\rho^2}+\frac{e^{-6 v}}{\rho^4}\Bigg[2\nabla Y^t\text{adj}(\bar{\lambda'})\nabla{\bar{Y}}+\nabla \bar{Y}^t\text{adj}(\lambda')\nabla{\bar{Y}}\nonumber\\
&-&6\bar{v}\nabla Y^t\text{adj}(\bar{\lambda'})\nabla Y-12\bar{v}\nabla Y^t\text{adj}(\lambda')\nabla\bar{Y}+18\bar{v}^2\nabla Y^t\text{adj}(\lambda')\nabla Y\Bigg]
\Bigg)\,d\Sigma
\end{eqnarray}}
Note that the integrand of the functional $\mathcal{M}$ is singular at $\rho=0$. However, we have defined the Banach space $B$ only for functions $\bar{Y}$ and $\bar{\lambda}'$ with support in $\Omega_{\rho_0}$. Therefore, the domain of integration of the terms in which $\nabla \bar{Y}$ and $\nabla\bar{\lambda'}$ appear are in fact $\Omega_{\rho_0}$ and hence the integrand is regular for those terms.

We  now introduce axillary Hilbert spaces $\mathcal{H}_i$, which is defined in terms of the weighted Sobolev spaces
\begin{eqnarray}
\norm{\bar{v}}^2_{\mathcal{H}_1}&=&\int_{\mathbb{R}^3}\abs{\nabla\bar{v}}^2 r^{-2}\td\Sigma+\int_{\mathbb{R}^3}\abs{\bar{v}}^2r^{-4}\td\Sigma\\
\norm{\bar{\lambda}'}^2_{\mathcal{H}_2}&=&\int_{\Omega_{\rho_0}}\abs{\nabla\bar{\lambda}'}^2\rho^{-2}\td\Sigma+\int_{\Omega_{\rho_0}}\abs{\bar{\lambda}'}^2\rho^{-4}\td\Sigma\\
\norm{\bar{Y}}^2_{\mathcal{H}_3}&=&\int_{\Omega_{\rho_0}}\abs{\nabla\bar{Y}}^2\rho^{-2}\td\Sigma+\int_{\Omega_{\rho_0}}\abs{\bar{Y}}^2\rho^{-4}\td\Sigma
\end{eqnarray}
and its corresponding inner products.  The following auxiliary Hilbert space for $\phi$ with norm defined by
\begin{equation}
\norm{\varphi}_{\mathcal{H}}^{2}=\norm{\bar{v}}^2_{\mathcal{H}_1}+\norm{\bar{\lambda}}^2_{\mathcal{H}_2}+\norm{\bar{Y}}^2_{\mathcal{H}_3},
\end{equation}
with its corresponding inner product. We have $B\subset \mathcal{H}$ and the following P\'oincare inequalities
\begin{lemma}\label{Poincare}
Let $\varphi\in\mathcal{H}$ and $\delta\neq 0$ is a real number . Then
\begin{enumerate}[(a)]
\item $\abs{\delta}^{-2}\int_{\mathbb{R}^3}\abs{\nabla\bar{v}}^2r^{-2\delta-1}\td\Sigma\geq
\int_{\mathbb{R}^3}\abs{\bar{v}}^2r^{-2\delta-3}\td\Sigma$
\item $\abs{\delta}^{-2}\int_{\Omega_{\rho_0}}\abs{\nabla\bar{\lambda}'}^2\rho^{-2\delta}\td\Sigma\geq
\int_{\Omega_{\rho_0}}\abs{\bar{\lambda}'}^2\rho^{-2\delta-2}\td\Sigma$
\item $2\abs{\delta}^{-2}\int_{\Omega_{\rho_0}}\nabla\bar{Y}^t\nabla\bar{Y}\rho^{-3\delta}\td\Sigma\geq
3\int_{\Omega_{\rho_0}}\bar{Y}^t\bar{Y}\rho^{-3\delta-2}\td\Sigma$
\end{enumerate}
 
\end{lemma}
\begin{proof}
\begin{enumerate}[(a)]
\item The proof of this part is similar to Theorem 1.3 of \cite{bartnik1986mass}. 
\item The proof of part (b) is as following. We know for any symmetric matrices $\bar{\lambda}$ we have
\begin{equation}
\abs{\bar{\lambda}'}^2=\bar{\lambda}'^2_{11}+\bar{\lambda}'^2_{22}+2\bar{\lambda}'^2_{12}
\end{equation}
Let $\Delta_3$ be Laplace operator respect to $\delta_3$ on $\mathbb{R}^3$. 
\begin{equation}
\Delta_3(\ln \rho)=0
\end{equation}
Then for each one of these functions, $\bar{\lambda}'_{ij}$ and by integrating over $\Omega_{\rho_0}$ and integrating by parts,
\begin{equation}
\int_{\Omega_{\rho_0}}\nabla\left(\rho^{-2\delta}\bar{\lambda}'^2_{ij}\right)\nabla\left(\ln\rho\right)\td\Sigma=0
\end{equation} 
Now if we expand the derivatives in the integrand and use H\"older inequality  we have
\begin{equation}
\abs{\delta}^{-2}\int_{\Omega_{\rho_0}}\abs{\nabla\bar{\lambda'}_{ij}}^2\rho^{-2\delta}\td\Sigma\geq
\int_{\Omega_{\rho_0}}\abs{\bar{\lambda}'_{ij}}^2\rho^{-2\delta-2}\td\Sigma
\end{equation}
Then we have the following inequality
\begin{equation}
\abs{\delta}^{-2}\int_{\Omega_{\rho_0}}\abs{\nabla\bar{\lambda}'}^2\rho^{-2\delta}\td\Sigma\geq
\int_{\Omega_{\rho_0}}\abs{\bar{\lambda}'}^2\rho^{-2\delta-2}\td\Sigma
\end{equation}
\item Proof is similar to part (b).
\end{enumerate}
\end{proof}

\begin{lemma}\label{uniformlemma} Let $\varphi\in B$ and $0<t<1$, then
\begin{enumerate}[(a)]
\item The function $\mathcal{E}_{\varphi}(t)$  is $C^2$ in the $t$ variable.
\item For every $\epsilon>0$ there exist $\eta(\epsilon)$ such that for
$\norm{\varphi}_{B}< \eta(\epsilon)$ we have
\begin{equation}\label{uc}
\abs{\mathcal{E}''_{\varphi}(t)-\mathcal{E}''_{\varphi}(0)}\leq\epsilon\norm{\varphi}^2_{\mathcal{H}}
\end{equation}
\end{enumerate}
\end{lemma}
\begin{proof}
\begin{enumerate}[(a)]
\item
To show $\mathcal{E}_{\varphi}(t)$  is $C^2$ it is enough to to show the third derivatives exists for all $t$. First we have
{\small\begin{eqnarray}
\mathcal{E}'''_{\varphi}(t)&=&\frac{1}{8}\int_{\mathbb{R}^3}\frac{e^{-6 v}}{\rho^4}\Bigg(3\nabla\bar{Y}^t\text{adj}(\bar{\lambda'})\nabla{{Y}}
-42\bar{v}\nabla\bar{Y}^t\text{adj}(\bar{\lambda'})\nabla{\bar{Y}}
-12\bar{v}\nabla\bar{Y}^t\text{adj}(\lambda')\nabla{\bar{Y}}\nonumber\\
&+&108\bar{v}^2\nabla{Y}^t\text{adj}(\bar{\lambda'})\nabla{ Y}
+144\bar{v}^2\nabla{Y}^t\text{adj}(\lambda')\nabla{\bar{Y}}
-216\bar{v}^3\nabla{Y}^t\text{adj}(\lambda')\nabla{ Y}
\Bigg)\,\td\Sigma\nonumber
\end{eqnarray}}
Note $\nabla\bar{Y}^i$  and $\bar{\lambda}'$ have compact support in $\Omega_{\rho_0}$. Therefore, by parts 1 and 2 of Definition \ref{Def1} and relation $\text{adj}\bar{\lambda}'=-\frac{1}{\rho^2}\text{adj}{\lambda}'_0\bar{\lambda}'\text{adj}\lambda'_0$ and $\det\bar{\lambda}'=0$ it is straightforward but tedious to show that all terms are bounded by the norm $B$. The only term with different domain is  
\begin{equation}
-\frac{216\bar{v}^3}{X_0}\nabla{Y}_0^t\lambda^{-1}_0\nabla{Y}_0
\end{equation}
which is bounded on $\mathbb{R}^3$ by part 1 of  Definition \ref{Def1}. Then $\mathcal{E}_{\varphi}(t)$ is $C^2$.
\item First by integrand of $\mathcal{E}''_{\varphi}(t)$ we have
\begin{eqnarray}
\mathcal{E}''_{\varphi}(t)-\mathcal{E}''_{\varphi}(0)=\int_{\mathbb{R}^3}\left(A_1|_{0}^t+...+A_6|_{0}^t\right)\,\td\Sigma
\end{eqnarray}
where 
\begin{eqnarray}
A_1&=&18\frac{e^{-6v}\bar{v}^2}{\rho^4}\nabla Y_0^t\text{adj}\lambda'_0\nabla Y_0\qquad A_2=\frac{e^{-6v}}{\rho^4}(18\bar{v}^2t-6\bar{v})\nabla {Y}^t_0\text{adj}\bar{\lambda}'\nabla Y_0\nonumber\\
A_3&=&\frac{e^{-6v}}{\rho^4}(36\bar{v}^2t-12\bar{v})\nabla \bar{Y}^t\text{adj}{\lambda}'_0\nabla Y_0\qquad
A_4=\frac{e^{-6v}}{\rho^4}(18\bar{v}^2t^2-12\bar{v}t+1)\nabla \bar{Y}^t\text{adj}\lambda'_0\nabla\bar{Y}\nonumber\\
A_5&=&\frac{e^{-6v}}{\rho^4}(36\bar{v}^2t^2-24\bar{v}t^2+2)\nabla \bar{Y}^t\text{adj}\bar{\lambda}'\nabla{Y}_0\qquad
A_6=\frac{e^{-6v}}{\rho^4}(18\bar{v}^2t^3-18\bar{v}t^2+3t)\nabla \bar{Y}^t\text{adj}\bar{\lambda}'\nabla\bar{Y}\nonumber
\end{eqnarray} 
All of these terms satisfy \eqref{uc} by similar steps as in \cite{dain2006proof}. We will explicitly give the proof for $A_1,A_2,A_3$ as the arguments are similar but tedious.  First we have
\begin{gather}\label{barv}
\abs{\bar{v}}\leq \sigma^{\beta}\norm{\bar{v}}_{C^1_{\beta}(\mathbb{R}^3)}\leq \norm{\bar{v}}_{C^1_{\beta}(\mathbb{R}^3)}\leq \norm{\varphi}_{B}\leq \eta
\end{gather}
By part (1) of Definition \ref{Def1}  we have 
\begin{eqnarray}
\int_{\mathbb{R}^3}A_1|_{0}^t\td\Sigma_0&=&\int_{\mathbb{R}^3}18\bar{v}^2\frac{\nabla Y_0^t\lambda_0^{-1}\nabla Y_0}{X_0}\left[e^{-6 t\bar{v}}-1\right]\td\Sigma\nonumber\\
&\leq&18 C\left[e^{6\eta}-1\right]\int_{\mathbb{R}^3}\bar{v}^2r^{-4}\td\Sigma \nonumber\\
&\leq& 18 C\left[e^{6\eta}-1\right]\norm{\bar{v}}^2_{\mathcal{H}_1}
\leq 18 C\left[e^{6\eta}-1\right]\norm{\varphi}^2_{\mathcal{H}}
\end{eqnarray}
Now First we write $A_2=B_1+B_2$ where
\begin{gather}
B_1=\frac{e^{-6v}}{\rho^4}18\bar{v}^2t\nabla {Y}^t_0\text{adj}\bar{\lambda}'\nabla Y_0\qquad B_2=-6\frac{e^{-6v_0}}{\rho^4}\bar{v}\nabla {Y}^t_0\text{adj}\bar{\lambda}'\nabla Y_0\left[e^{-6 t\bar{v}}-1\right]
\end{gather}
We will prove it for $B_1$ and $B_2$ is similar. We have
\begin{eqnarray}
\int_{\mathbb{R}^3}B_1\td\Sigma &=&-\int_{\Omega_{\rho_0}}\frac{e^{-6v}}{\rho^6}18\bar{v}^2t\nabla {Y}^t_0\text{adj}\lambda'_0\bar{\lambda}'\text{adj}\lambda'_0\nabla Y_0\td\Sigma\nonumber\\
&\leq&18e^{6\eta}\eta\int_{\Omega_{\rho_0}}\frac{e^{-6v_0}}{\rho^4}\abs{\bar{\lambda}'}\bar{v}\nabla {Y}^t_0(\text{adj}{\lambda}'_0)^2\nabla Y_0\td\Sigma\nonumber\\
&\leq&18C\eta e^{6\eta}\int_{\Omega_{\rho_0}}\abs{\bar{\lambda}'}\bar{v}\rho^{-1}r^{-2}\td\Sigma \nonumber\\
&\leq&18C\eta e^{6\eta}\norm{\bar{v}}_{\mathcal{H}_1}\norm{\bar{\lambda'}}_{\mathcal{H}_2}\leq 18C\eta e^{6\eta}\norm{\varphi}^2_{\mathcal{H}}
\end{eqnarray}
We used the identity $\text{adj}\bar{\lambda}'=-\frac{1}{\rho^2}\text{adj}{\lambda}'_0\bar{\lambda}'\text{adj}\lambda'_0$ in the first line.  The first inequality arise from \eqref{barv} and the matrix inequality $u^t A u \leq \abs{A} u^t u$ for any $2 \times 2$ matrix $A$. The second inequality is a consequence of parts (1) and (2) of Definition \ref{Def1}.  Finally, the third inequality follows from H\"older's inequality.

The term $A_3$ can be expressed as $A_3=B_3+B_4$ where
\begin{gather}
B_3=36\frac{e^{-6v}}{\rho^4}\bar{v}^2t\nabla \bar{Y}^t\text{adj}{\lambda}'_0\nabla Y_0\qquad B_4=-12\frac{e^{-6v_0}}{\rho^4}\bar{v}\nabla \bar{Y}^t\text{adj}{\lambda}'_0\nabla Y_0\left[e^{-6 t\bar{v}}-1\right]
\end{gather}
Then the bound of $B_3$ is 
\begin{eqnarray}
\int_{\mathbb{R}^3}B_3\td\Sigma
&\leq&36\eta e^{6\eta}\int_{\Omega_{\rho_0}}\frac{1}{X_0}\bar{v}\nabla \bar{Y}^t\lambda'^{-1}_0\nabla Y_0\td\Sigma \nonumber\\
&\leq&36\eta e^{6\eta}\int_{\Omega_{\rho_0}}\frac{\bar{v}}{X_0}\left(\nabla \bar{Y}^t\lambda'^{-1}_0\nabla \bar{Y}\right)^{1/2} \left(\nabla Y_0^t \lambda_0^{-1} \nabla Y_0 \right)^{1/2}\td\Sigma \nonumber \\
&\leq&36C\eta e^{6\eta}\left(\int_{\Omega_{\rho_0}}\rho^{-2}\nabla \bar{Y}^t\lambda'^{-1}_0\nabla \bar{Y}\td\Sigma \right)^{1/2}\left(\int_{\Omega_{\rho_0}}\bar{v}^2r^{-4}\td\Sigma \right)^{1/2}\nonumber\\
&\leq&36C\eta e^{6\eta}\norm{\bar{v}}_{\mathcal{H}_1}\norm{\bar{Y}}_{\mathcal{H}_3}
\leq 36C\eta e^{6\eta}\norm{\varphi}^2_{\mathcal{H}}  
\end{eqnarray} The first inequality uses \eqref{barv}.  We know $\lambda_0^{-1}$ is a positive definite symmetric matrix. Thus it has a square root matrix $\lambda_0^{-1/2}$, that is $\lambda_0^{-1}=\left(\lambda_0^{-1/2}\right)^2$.  Then the integrand in the first line is equal to $X_0^{-1} \bar{v} u^t w$ where $u^t = \nabla \bar{Y}^t \lambda_0^{-1/2}$ and $w = \lambda_0^{-1/2}\nabla Y_0$.  Since $u^t w \leq (u^t u)^{1/2} (w^t w)^{1/2}$ we have the second inequality.  The third inequality follows from H\"older's inequality and parts (1) and  (2) of Definition 1.  The fourth inequality is by the definition of norm.
$B_4$ is exactly similar to $B_3$.
\end{enumerate}
\end{proof}
\section{Local minima of $\mathcal{E}_\varphi(t)$ }\label{Carter}
In this section we  first derive a five-dimensional version of Carter's identity and show its relation with the second variation $\mathcal{E}''_{\varphi}(t)$. Assume we are in a five dimensional vacuum spacetime with isometry group $\mathbb{R}\times U(1)^2$.  The field equations can be expressed simply as the conservation of a current (see \cite{figueras2010uniqueness} for details). 
 \begin{equation}\label{EOM}
 \nabla\cdot J=\nabla\cdot\left(\rho\,\Phi^{-1}\nabla\Phi\right)=0
 \end{equation}
 where
 \begin{equation} \Phi \equiv
 \Phi(X,Y,\lambda)=\begin{pmatrix}
   \frac{1}{X} & -\frac{Y^t}{X} \\
  -\frac{Y}{X} & \lambda+\frac{YY^t}{X}
  \end{pmatrix}
 \end{equation}
and $\det\Phi=1$, $\lambda$ is a positive definite $2\times 2$ symmetric matrix with $\det\lambda=X$ and $Y$ is a column vector. One can derive the Mazur identity (for a detailed discussion see \cite{carter1985bunting}) for two matrices $\Phi_{[1]}$ and $\Phi_{[2]}$ (not necessarily solutions) with corresponding currents $J_{[1]}, J_{[2]}$ 
\begin{equation} \label{Mazur}
\Delta \Psi-\text{Tr}\left(\Phi_{[2]}\left(\nabla\cdot\mathring{J} \right)\Phi_{[1]}^{-1}\right)=\frac{1}{\rho^2}\text{Tr}\left(\mathring{J}^t\Phi_{[2]}\mathring{J}\Phi_{[1]}^{-1}\right)
\end{equation}
where $\Delta$ is Laplace operator with respect to flat metric $\delta_3$ and
\begin{eqnarray}
\Psi&=&\text{Tr}\left(\Phi_{[2]}\Phi_{[1]}^{-1}-I\right)\qquad \mathring{J}=J_{[2]}-J_{[1]}
\end{eqnarray}   Note that this identity holds quite generally for any field theory which can be derived from a positive definite action with Lagrangian of the form $L \sim \Tr (\Phi^{-1} \td \Phi)^2$. The linearized version of this identity in four dimensions  was originally found by  Carter  \cite{carter1971axisymmetric} and plays an important role in geometric inequalities in 3+1 dimensional spacetime\cite{dain2006proof,dain2008proof,dain2011area,dain2011areacharge}.  We will now derive a generalization of this identity for five dimensions.  Assume we have  $\Phi_{[1]}(X,Y,\lambda)$ and $\Phi_{[2]}(X_2,Y_2,\lambda_2)$ related by
\begin{gather}\label{fieldeqns}
\begin{aligned}X_2&=X+s\dot{X}\qquad &&Y_2 =Y+s \dot{Y}  \qquad \lambda_2=\lambda+s\dot{\lambda} \\
G_{\lambda_2}&=G_{\lambda}+s \dot{G}_{\lambda},\qquad  &&G_{X_2}=G_X+s \dot{G}_X 
\end{aligned}
\end{gather}
The overdot $\dot{}$ represents the linear order of expansion or first variation with respect to $s$ (when taking variations of the products of several terms, we use the notation $\delta$ instead of dot for convenience of notation). Then \eqref{Mazur} implies, to lowest order in $s$, that
\begin{eqnarray}
&&\Delta\left(\frac{\dot{Y}^t\lambda^{-1}\dot{Y}}{X}+\frac{\dot{X}^2}{X^2}\right)\nonumber\\
&+&\frac{\dot{Y}^t\lambda^{-1}\dot{Y}}{X}G_X-\frac{\dot{X}}{X}\dot{G}_X-2\frac{\dot{X}}{X}G_Y^t\dot{Y}-2\dot{Y}^t\lambda^{-1}\dot{\lambda}G_Y+\frac{\dot{Y}^t\lambda^{-1}G^t_{\lambda}\dot{Y}}{X}-\text{Tr}\left(\lambda^{-1}\dot{\lambda}\dot{G}^t_{\lambda}\right)-2\dot{G}_Y^t\dot{Y}\nonumber\\
&=&\left(\nabla\left(\frac{\dot{X}}{X}\right)+\frac{\dot{Y}^t\lambda^{-1}\nabla Y}{X}\right)^2+X\left(\dot{U}_2^t\lambda\dot{U}_2+\nabla U_1^t\lambda\nabla U_1\right)+\text{Tr}\bigg[\left(\nabla\left(\dot{\lambda}\lambda^{-1}\right)+\frac{\nabla Y\dot{Y}^t\lambda^{-1}}{X}\right)^2\bigg]\nonumber\\
\end{eqnarray}
where
\begin{equation}
U_1\equiv\frac{\lambda^{-1}\dot{Y}}{X}\qquad U_2\equiv\frac{\lambda^{-1}\nabla{Y}}{X}
\end{equation} This is the five-dimensional extension of Carter's identity which appeared in \cite{carter1971axisymmetric} .
Now if we consider our parametrization of data with relations \eqref{relation1} and \eqref{relation2} we have
\begin{equation}
\dot{X}=4\bar{v}X,\qquad \dot{\lambda}=\bar{\lambda}=2\bar{v}\lambda+\lambda\lambda'^{-1}\bar{\lambda}',\qquad \dot{Y}=\bar{Y}
\end{equation}
Thus
\begin{equation}
\lambda^{-1}\dot{\lambda}=2\bar{v}I+\lambda'^{-1}\bar{\lambda}'\label{lambdaandprime}
\end{equation}
since  $\text{Tr}\left(\lambda'^{-1}\bar{\lambda}'\right)= \delta \det \lambda' / \det \lambda' = 0$ we have $\text{Tr}\left(\lambda^{-1}\bar{\lambda}\right)=4\bar{v}$. Then the following identity holds for arbitrary $v$, $\bar{v}$,$Y$, $\bar{Y}$,$\lambda$,$\bar{\lambda}$  will be 
\begin{eqnarray}\label{Carteridentity}
&&\Delta\left(\frac{\bar{Y}^t\lambda^{-1}\bar{Y}}{X}+16\bar{v}^2\right)\nonumber\\
&+&\frac{\bar{Y}^t\lambda^{-1}\bar{Y}}{X}G_X-4\bar{v}\dot{G}_X-8\bar{v}G_Y^t\bar{Y}-2\bar{Y}^t\lambda^{-1}\bar{\lambda}G_Y+\frac{\bar{Y}^t\lambda^{-1}G^t_{\lambda}\bar{Y}}{X}-\text{Tr}\left(\lambda^{-1}\bar{\lambda}\dot{G}^t_{\lambda}\right)-2\dot{G}_Y^t\bar{Y}\nonumber\\
&=&F(t)
\end{eqnarray}
where $G_X$, $G_Y$, and $G_{\lambda}$ defined in \eqref{fieldeqns} and 
\begin{eqnarray}
F(t)&=&\left(4\nabla\bar{v}+\frac{\bar{Y}^t\lambda^{-1}\nabla Y}{X}\right)^2+X\left(\dot{U}_2^t\lambda\dot{U}_2+\nabla U_1^t\lambda\nabla U_1\right)+\text{Tr}\bigg[\left(\nabla\left(\bar{\lambda}\lambda^{-1}\right)+\frac{\nabla Y\bar{Y}^t\lambda^{-1}}{X}\right)^2\bigg]\nonumber\\
\dot{G}_X&=&4\Delta_3\bar{v}+\frac{e^{-6v}}{\rho^4}\left\{2\nabla{\bar{Y}}^t\text{adj}\lambda'\nabla Y+\nabla{{Y}}^t\text{adj}\bar{\lambda'}\nabla Y-6\bar{v}\nabla{{Y}}^t\text{adj}\lambda'\nabla Y\right\}\nonumber\\
\dot{G}_{\lambda}&=&2\Delta_3\bar{v} I+\nabla\cdot\delta\left(\lambda ^{'-1}\nabla\lambda'\right)+\frac{e^{-6v}}{\rho^4}\left\{2\text{adj}\lambda'\nabla Y\cdot\nabla \bar{Y}^t+\text{adj}\bar{\lambda}'\nabla Y\cdot\nabla {Y}^t-6\bar{v}\text{adj}\lambda'\nabla Y\cdot\nabla{Y}^t\right\}\nonumber\\
\dot{G}_{Y}&=&\nabla\cdot \left(\frac{e^{-6v}}{\rho^4}\left\{\text{adj}\lambda'\nabla\bar{Y}+\text{adj}\bar{\lambda}'\nabla{Y}-6\bar{v}\text{adj}\lambda'\nabla{Y}\right\}\right)
\end{eqnarray} The identity \eqref{Carteridentity} can be verified directly.
Assume $\varphi\in B$ then after a tedious calculation involving  repeated integration by parts we have the remarkable relation
\begin{equation}\label{secondvariation}
\int_{\mathbb{R}^3}\left(-4\bar{v}\dot{G}_X-\text{Tr}\left(\lambda^{-1}\bar{\lambda}\dot{G}_{\lambda}\right)-2\dot{G}_Y^t\bar{Y}\right)\td\Sigma=16\mathcal{E}''_{\varphi}(t)
\end{equation} 
Thus if $t=0$, the field equations $G_X(0)=G_{\lambda}(0)=G_Y(0)=0$ hold and  we have from \eqref{Carteridentity} (the integral over the divergence term vanishes by our boundary conditions) 
\begin{equation}
\mathcal{E}''_{\varphi}(0)=\frac{1}{16}\int_{\mathbb{R}^3}F(0)\td\Sigma\geq 0
\end{equation}
where
\begin{eqnarray}
F(0)&=&\left(4\nabla\bar{v}+\frac{\bar{Y}^t\lambda^{-1}_0\nabla Y_0}{X_0}\right)^2+X_0\left(\dot{U}_2^t\lambda\dot{U}_2+\nabla U_1^t\lambda\nabla U_1\right)+\text{Tr}\bigg[\left(\nabla\left(\bar{\lambda}\lambda^{-1}_0\right)+\frac{\nabla Y_0\bar{Y}^t\lambda^{-1}_0}{X_0}\right)^2\bigg]\nonumber\\
&\geq& X_0\nabla U_1^t\lambda_0\nabla U_1\label{F1}
\end{eqnarray}
Now if $\mathcal{E}''_{\varphi}(0)=0$, then $F(0)=0$. Therefore, by inequality \eqref{F1} we have $\nabla U_1=0$. Also, since $\varphi\in B$, we have $\bar{Y}=0$. Therefore, by $F=0$ and $\bar{Y}=0$ we have $\bar{v}=0$ and $\bar{\lambda}=0$.  This is, however, not sufficient to prove that the extreme data $u_0$ is  a \emph{strict} local minimum. For this one needs a stronger positivity result on $\mathcal{E}''_{\varphi}(0)$  (see for example, Theorem 40.B of \cite{zeidler1989nonlinear}) which we now demonstrate.

Firstly, we prove a coercive condition required for $u_0$ to be a local minimum. We note the  identity (this arises in the proof of \eqref{secondvariation})
\begin{eqnarray}
\int_{\Omega_{\rho_0}}2\rho^{-2}\text{Tr}\left(\lambda'^{-1}\nabla \lambda'\text{adj}\bar{\lambda'}\nabla \bar{\lambda'}\right)\,\td\Sigma=-\int_{\Omega_{\rho_0}}\left(\text{Tr}\left[\bar{\lambda}'\nabla\left(\lambda'^{-1}\right)\right]\right)^2\,\td\Sigma\label{Crazyidentity}
\end{eqnarray}
\begin{lemma}\label{coercive}
There exist $\mu>0$ such that for all $\varphi\in B$ we have
\begin{equation}
\mathcal{E}''_{\varphi}(0)\geq\mu\norm{\varphi}^2_{\mathcal{H}}\label{Coercivein}
\end{equation}
\end{lemma} 
\begin{proof}
Let $\varphi\in B$. Note that $\mathcal{E}''_{\varphi}(0)$ defines a bilinear form 
\begin{equation}
a(\varphi,\varphi)\equiv\mathcal{E}''_{\varphi}(0)=\int_{\mathbb{R}^3}F(0)\td\Sigma
\end{equation}
as function of $\varphi$. The inequality \eqref{Coercivein} is equivalent to the following variational problem 
\begin{equation}
\mu=\inf_{\varphi\in B,\norm{\varphi}^2_{\mathcal{H}}=1}a(\varphi,\varphi)
\end{equation}
Since $a(\varphi,\varphi)$ is positive definite, we have $\mu\geq 0$. Now we prove $\mu > 0$. Assume $\mu=0$, then there exists a sequence $\{\varphi_n\}$ such that
\begin{equation}
\norm{\varphi_n}^2_{\mathcal{H}}=1\qquad \text{for all $n$}
\end{equation}
and 
\begin{equation}
\lim_{n\to\infty}a(\varphi_n,\varphi_n)=0
\end{equation}
Then we have
\begin{eqnarray}
0&=&\lim_{n\to\infty}a(\varphi_n,\varphi_n)=\lim_{n\to\infty}\int_{\mathbb{R}^3}F(0)\td\Sigma\nonumber\\
&\geq&\lim_{n\to\infty} \int_{\Omega_{\rho_0}}X_0\nabla U_1^t\lambda_0\nabla U_1\td\Sigma 
\geq C_1\lim_{n\to\infty}\int_{\Omega_{\rho_0}}\rho^3 \nabla U_1^t\nabla U_1\td\Sigma\nonumber\\
& \geq& \frac{3C_1}{2}\lim_{n\to\infty}\int_{\Omega_{\rho_0}}\rho U_1^t U_1\td\Sigma
\geq\frac{3C_1C_3}{2C'}\lim_{n\to\infty}\int_{\Omega_{\rho_0}}\rho^{-4}\bar{Y}^t_n\lambda^{-1}_0\bar{Y}_n\td\Sigma
\end{eqnarray}
In first inequality we used \eqref{F1}. The second follows from part 2 and 3 of Definition \ref{Def1} . Third inequality follows from  Lemma \ref{Poincare}-(c). Fourth inequality follows from part 3 of Definition \ref{Def1}. Therefore, 
\begin{equation}
\lim_{n\to\infty}\int_{\Omega_{\rho_0}}\rho^{-4}\bar{Y}^t_n\lambda^{-1}_0\bar{Y}_n\td\Sigma=0\label{firstY}
\end{equation}
Next we establish some inequalities.  First rewrite $F(0)$ in the form 
\begin{eqnarray}
F(0)&=&\left(4\nabla\bar{v}_n+\frac{\bar{Y}^t_n\lambda^{-1}_0\nabla Y_0}{X_0}\right)^2+2A_1^t\lambda_0 A_1+2A_2^t\lambda_0 A_2+\text{Tr}\bigg[\left(\nabla\left(\bar{\lambda}_n\lambda^{-1}_0\right)+\frac{\nabla Y_0\bar{Y}^t_n\lambda^{-1}_0}{X_0}\right)^2\bigg]\nonumber
\end{eqnarray}
where 
\begin{equation}
A_1 = \frac{\sqrt{X_0}}{2} \left[B_I + B_{II} + B_{III}\right], \qquad A_2 = \frac{\sqrt{X_0}}{2} \left[B_{II} - B_I\right]
\end{equation} and
\begin{eqnarray}
B_I &=& {\frac{\lambda^{-1}_0\nabla\lambda_0\lambda^{-1}_0\bar{Y}_n}{X}+\frac{\nabla{X}_0}{X^2_0}\lambda^{-1}_0\bar{Y}_n}, \qquad B_{II} = \frac{\lambda^{-1}_0\bar{\lambda}_n\lambda^{-1}_0\nabla{Y}_0}{X_0}+\frac{\bar{X}}{X^2_0}\lambda^{-1}_0\nabla Y_0 \nonumber \\ B_{III} &=& 2\frac{\lambda^{-1}_0\nabla\bar{Y}}{X_0} \; .
\end{eqnarray}
Then we have the following inequality
\begin{equation}
a(\varphi_n,\varphi_n)+\int_{\Omega_{\rho_0}}2B_I^t\lambda_0B_I\,\td\Sigma\geq\int_{\Omega_{\rho_0}}\frac{1}{4}B_{III}^t\lambda_0B_{III}
\,\td\Sigma\label{ineq1}
\end{equation}
where $B_I$ can be written as
\begin{equation}
B_I=\frac{\lambda^{-1}_0}{\sqrt{X_0}}\left(\nabla\lambda_0\lambda^{-1}_0+\frac{\nabla{X_0}}{X_0}I_{2\times 2}\right)\bar{Y}_n=\frac{\lambda^{-1}_0}{\sqrt{X_0}}M\bar{Y}_n \;.
\end{equation}
By part 4 of Definition \ref{Def1} we have
\begin{equation}
\abs{M}^2\leq 2\abs{\nabla\lambda_0\lambda^{-1}_0}^2+2\abs{\nabla\ln X_0}^2\leq C\rho^{-2}
\end{equation}
and we have
\begin{eqnarray}
\int_{\Omega_{\rho_0}}2B_I^t\lambda_0B_I\,\td\Sigma&\leq&\int_{\Omega_{\rho_0}}\frac{2}{X_0}\abs{M}^2\bar{Y}^t_n\lambda^{-1}_0\bar{Y}_n\,\td\Sigma\leq  2C\int_{\Omega_{\rho_0}}\rho^{-4}\bar{Y}^t_n\lambda^{-1}_0\bar{Y}_n\,\td\Sigma\label{ineq2}
\end{eqnarray}
Then by inequities \eqref{ineq1} and \eqref{ineq2} we have
\begin{equation}
a(\varphi_n,\varphi_n)+4\int_{\Omega_{\rho_0}}\rho^{-4}\bar{Y}^t_n\lambda^{-1}_0\bar{Y}_n\,\td\Sigma\geq \frac{1}{4}\int_{\Omega_{\rho_0}}\rho^{-2}\nabla\bar{Y}^t_n\lambda^{-1}_0\nabla\bar{Y}_n\,\td\Sigma
\,\td\Sigma
\end{equation}
Now we take the limit of above equation and use the equation \eqref{firstY} to find
\begin{equation}
\lim_{n\to\infty}\int_{\Omega_{\rho_0}}\rho^{-2}\nabla\bar{Y}^t_n\lambda^{-1}_0\nabla\bar{Y}_n\td\Sigma=0
\end{equation}
Thus 
\begin{equation}\label{Ynvanish}
\lim_{n\to\infty}\norm{\bar{Y}_n}_{\mathcal{H}_3}=0
\end{equation}
Now we look at the first term in $F(0)$. Then
\begin{equation}
a(\varphi_n,\varphi_n)+\int_{\Omega_{\rho_0}}\left(\frac{\bar{Y}^t_n\lambda^{-1}_0\nabla Y_0}{X_0}\right)^2\,\td\Sigma\geq 8\int_{\Omega_{\rho_0}}\left(\nabla\bar{v}_n\right)^2\,\td\Sigma \label{ineq4}
\end{equation}
Since $\lambda_0$ is a positive definite symmetric metric it has unique square root $\lambda^{1/2}_0$. Now if we set $u=\lambda_0^{-1/2}\bar{Y}$ and $w=\lambda_0^{-1/2}\nabla Y_0$  we have
\begin{eqnarray}
\int_{\Omega_{\rho_0}}\left(\frac{\bar{Y}^t_n\lambda^{-1}_0\nabla Y_0}{X_0}\right)^2\,\td\Sigma &\leq&\int_{\Omega_{\rho_0}}\left(\frac{\bar{Y}^t_n\lambda^{-1}_0\bar{Y}_n}{X_0}\right)\left(\frac{\nabla{Y}^t_0\lambda^{-1}_0\nabla{Y}_0}{X_0}\right)\,\td\Sigma \nonumber\\
&\leq&C\int_{\Omega_{\rho_0}}\rho^{-2}\bar{Y}^t_n\lambda^{-1}_0\bar{Y}_nr^{-4}\,\td\Sigma\nonumber\\
&\leq&C\int_{\Omega_{\rho_0}}\rho^{-4}\bar{Y}^t_n\lambda^{-1}_0\bar{Y}_n\,\td\Sigma \label{ineq3}
\end{eqnarray}
The first inequality follows from the Cauchy-Schwarz inequality $u^tw\leq (u^tu)^{1/2}(w^tw)^{1/2}$. Second inequality is by part 1 and 3 of Definition \ref{Def1}. The third inequality is by the fact $\rho\leq r^2$.  Then by inequality \eqref{ineq3} and \eqref{ineq4} we have
\begin{equation}
a(\varphi_n,\varphi_n)+C\int_{\Omega_{\rho_0}}\rho^{-4}\bar{Y}^t_n\lambda^{-1}_0\bar{Y}_n\,\td\Sigma \geq 8\int_{\mathbb{R}^3}\left(\nabla\bar{v}_n\right)^2\,\td\Sigma  \geq 8\int_{\mathbb{R}^3}\left(\nabla\bar{v}_n\right)^2 r^{-2}\,\td\Sigma\label{ineq5}
\end{equation} the last inequality is by Theorem 1.2-(i) of \cite{bartnik1986mass}.
Now if we take the limit of inequality \eqref{ineq5} and by the fact the right hand side is zero by \eqref{firstY}, we have
\begin{equation}
\lim_{n\to \infty}\int_{\mathbb{R}^3}\left(\nabla\bar{v}_n\right)^2 r^{-2}\,\td\Sigma=0\label{vinfty}
\end{equation}
Thus by Lemma \ref{Poincare}-(a)  we have
\begin{equation}\label{vnvanish}
\lim_{n\to \infty}\norm{\bar{v}_n}_{\mathcal{H}_1}=0
\end{equation}
Now we consider the last term of $F(0)$. We have the following inequality
\begin{equation}
a(\varphi_n,\varphi_n)+\int_{\Omega_{\rho_0}}\text{Tr}\bigg[\left(\frac{\nabla Y_0\bar{Y}^t_n\lambda^{-1}_0}{X_0}\right)^2\bigg]\,\td\Sigma \geq \frac{1}{2}\int_{\Omega_{\rho_0}}\text{Tr}\bigg[\left(\nabla\left(\bar{\lambda}_n\lambda^{-1}_0\right)\right)^2\bigg]\,\td\Sigma \label{ineq8}
\end{equation}
The integrand of the second term on the left hand side has vanishing determinant since $\det\left(\nabla Y_0\bar{Y}^t_n\lambda^{-1}_0\right)=\frac{\det\left(\nabla Y_0\bar{Y}^t_n\right)}{\rho^2}=0$. Thus by the matrix identity $\text{Tr}(A^2)=\left(\text{Tr}A\right)^2-2\det A$ and inequality \eqref{ineq3} we have
\begin{eqnarray}
\int_{\Omega_{\rho_0}}\text{Tr}\bigg[\left(\frac{\nabla Y_0\bar{Y}^t_n\lambda^{-1}_0}{X_0}\right)^2\bigg]\,\td\Sigma &\leq&C\int_{\Omega_{\rho_0}}\rho^{-4}\bar{Y}^t_n\lambda^{-1}_0\bar{Y}_n\,\td\Sigma \label{ineq7}
\end{eqnarray}
By relation \eqref{lambdaandprime}
the right hand side expands  
\begin{eqnarray}
\text{Tr}\bigg[\left(\nabla\left(\bar{\lambda}_n\lambda^{-1}_0\right)\right)^2\bigg]
&=&2\left[\nabla\bar{v}_n\right]^2+\text{Tr}\left[\left(\nabla\bar{\lambda}'_n\lambda'^{-1}_0\right)^2\right]+\text{Tr}\left[\left(\bar{\lambda}'_n\nabla\left(\lambda'^{-1}_0\right)\right)^2\right]\nonumber\\&+&2\text{Tr}\left[\nabla\bar{\lambda}'_n\left(\frac{\text{adj}\bar{\lambda}'_n}{\rho^2}\right)\nabla\lambda'_0\lambda'^{-1}_0\right]
\end{eqnarray}
By integration we have
\begin{eqnarray}
\int_{\Omega_{\rho_0}}\text{Tr}\bigg[\left(\nabla\left(\bar{\lambda}_n\lambda^{-1}_0\right)\right)^2\bigg]\,\td\Sigma&=&\int_{\mathbb{R}^3}2\left[\nabla\bar{v}_n\right]^2\,\td\Sigma+\int_{\Omega_{\rho_0}}\text{Tr}\left[\left(\nabla\bar{\lambda}'_n\lambda'^{-1}_0\right)^2\right]\,\td\Sigma \nonumber\\
&\geq&\int_{\mathbb{R}^3}2\left[\nabla\bar{v}_n\right]^2\,\td\Sigma+C_1^2\int_{\Omega_{\rho_0}}\abs{\nabla\bar{\lambda}'_n}^2\rho^{-2}\,\td\Sigma\qquad\label{ineq6}
\end{eqnarray}
The equality is by identity \eqref{Crazyidentity}. The inequality is by part 2 of Definition \ref{Def1}. Then by substitution of inequalities \eqref{ineq6} and \eqref{ineq7} in \eqref{ineq8} we have
\begin{eqnarray}
a(\varphi_n,\varphi_n)+C\int_{\Omega_{\rho_0}}\rho^{-4}\bar{Y}^t_n\lambda^{-1}_0\bar{Y}_n\,\td\Sigma \geq \int_{\mathbb{R}^3}\left[\nabla\bar{v}_n\right]^2\,\td\Sigma+\frac{C_1^2}{2}\int_{\Omega_{\rho_0}}\abs{\nabla\bar{\lambda}'_n}^2\rho^{-2}\,\td\Sigma
\end{eqnarray}
Now if we take the limit from both side of this inequality and use equation \eqref{vinfty} we have
\begin{equation}
\lim_{n\to\infty}\int_{\Omega_{\rho_0}}\abs{\nabla\bar{\lambda}'_n}^2\rho^{-2}\,\td\Sigma=0
\end{equation}
Thus by Lemma \ref{Poincare}-(b) we have
\begin{equation}\label{lambdavan}
\lim_{n\to\infty}\norm{\bar{\lambda}'_n}_{\mathcal{H}_2}=0
\end{equation}
Thus \eqref{Ynvanish}, \eqref{vnvanish} and \eqref{lambdavan} contradict the fact that $\norm{\varphi_n}_{\mathcal{H}}=1$. Hence $\mu>0$.
\end{proof}

\section{Proof of Theorem \ref{main theorem}}\label{proof}
\begin{proof} The proof is straightforward and similar to the proof of theorem 1 of \cite{dain2006proof} and Chapter 40-B of \cite{zeidler1989nonlinear}.
\begin{enumerate}[(a)]
\item  We have proved in Lemma \ref{uniformlemma} that $\mathcal{E}''_{\varphi}(t)$ is $C^2$ with respect to $t$. Also by Taylor's theorem we have
\begin{equation}
\mathcal{M}(u_0+\varphi)-\mathcal{M}(u_0)=\mathcal{E}_{\varphi}(1)-\mathcal{E}_{\varphi}(0)=\frac{\mathcal{E}''_{\varphi}(t)}{2}\quad 0<t<1
\end{equation}
To prove this is positive we will show $\mathcal{E}''_{\varphi}(t)\geq 0$ and $\mathcal{E}''_{\varphi}(t)=0$ implies $\varphi=0$.
By Lemma \ref{uniformlemma}-(b) $\mathcal{E}''_{\varphi}(t)$ is uniformly continuous,  that is  for every $\epsilon>0$ there exist $\eta(\epsilon)$ such that the following inequality holds
\begin{equation}
\abs{\mathcal{E}''_{\varphi}(t)-\mathcal{E}''_{\varphi}(0)}\leq \epsilon \norm{\varphi}^2_{\mathcal{H}}
\end{equation}
for every $\norm{\varphi}_{\mathcal{H}}<\eta(\epsilon)$. From this inequality we have
\begin{equation}
\mathcal{E}''_{\varphi}(0)-\epsilon \norm{\varphi}^2_{\mathcal{H}}\leq \mathcal{E}''_{\varphi}(t)
\end{equation}
By Lemma \ref{coercive} we have
\begin{equation}
(\mu-\epsilon) \norm{\varphi}^2_{\mathcal{H}}\leq \mathcal{E}''_{\varphi}(t)
\end{equation}
Choosing $\eta(\epsilon)$ such that $0 < \epsilon<\mu$ the desired result follows.

\item Let $u= u_0 + \varphi$ be the associated $t-\phi^i$ symmetric part of the initial data set $(\Sigma,h,K)$ as in the statement of Theorem \ref{main theorem}. It was proved that the ADM mass of this data satisfies \cite{alaee2014mass}
\begin{equation}\label{ADMineq}
m\geq\mathcal{M}(u) = \mathcal{M}(u_0 + \varphi)
\end{equation}
Then by part (a) we have
\begin{equation}
\mathcal{M}(u_0+\varphi)> \mathcal{M}(u_0)
\end{equation}
for nonzero $\varphi$. Since $u_0$ is an extreme data, there exists a function $f$ such that $\mathcal{M}(u_0)=f(J_1,J_2)$. Thus
\begin{equation}
m \geq f(J_1,J_2)
\end{equation}

Clearly,  by definition if the initial data is extreme, then $m=f(J_1,J_2)$ . Conversely, suppose the mass $m$ of given initial data $(\Sigma,h,K)$ satisfies $m=f(J_1,J_2)=\mathcal{M}(u_0)$.  Hence $\varphi =0$ and $u= u_0$ and from \eqref{ADMineq} and Remark \ref{remark1} the initial data is extreme. Thus $m=f(J_1,J_2)$ if and only if  the data belongs to the extreme class.
\end{enumerate}
\end{proof}
\subsection*{Acknowledgments}
 We would like to thank S. Dain for comments on the use of the Carter identity in his article  \cite{dain2006proof} and also for clarifying the relationship between spacetime Weyl and quasi-isotropic coordinates for non-extreme data.  HKK also thanks James Lucietti for discussions concerning the uniqueness theorem for extreme black holes \cite{figueras2010uniqueness}.  We also would like to thank the referees for suggesting a number of improvements.  AA is partially supported by a graduate scholarship from Memorial University. HKK is supported by an NSERC Discovery Grant.  

\appendix
\numberwithin{equation}{section}

\bibliographystyle{unsrt}
\bibliographystyle{abbrv}  
\bibliography{masterfile}
       
\end{document}